\documentclass[10pt,a4paper]{article}

\usepackage[utf8]{inputenc}
\usepackage[T1]{fontenc}
\usepackage{microtype}
\usepackage{fourier}

\usepackage{tikz}
\usepackage[margin=10pt,font=small,labelfont=bf,labelsep=period]{caption}
\usepackage{amsmath}
\usepackage{amsthm}
\usepackage{amssymb}
\usepackage{mathscinet}
\usepackage{graphicx}
\numberwithin{equation}{section}

\usepackage[nottoc,notlot,notlof]{tocbibind}

\usepackage{aliascnt}
\usepackage[colorlinks,linkcolor=blue,citecolor=blue]{hyperref}

\theoremstyle{plain}

\newtheorem{theorem}{Theorem}[section]

\newaliascnt{lemma}{theorem}
\newtheorem{lemma}[lemma]{Lemma}
\aliascntresetthe{lemma}

\newaliascnt{corollary}{theorem}
\newtheorem{corollary}[corollary]{Corollary}
\aliascntresetthe{corollary}

\theoremstyle{definition}

\newaliascnt{definition}{theorem}
\newtheorem{definition}[definition]{Definition}
\aliascntresetthe{definition}

\newaliascnt{example}{theorem}

\aliascntresetthe{example}

\newaliascnt{remark}{theorem}
\newtheorem{remark}[remark]{Remark}
\aliascntresetthe{remark}

\newcommand{\R}{\mathbf{R}}
\newcommand{\C}{\mathbf{C}}

\newcommand{\N}{\mathbf{N}}
\renewcommand{\l}{\lambda}
\renewcommand{\epsilon}{\varepsilon}

\newcommand{\abs}[1]{\left\lvert #1 \right\rvert}
\newcommand{\avg}[1]{\bigl\langle #1 \bigr\rangle}

\DeclareMathOperator{\Tr}{Tr}
\title{The Calogero-Fran\c coise integrable system: algebraic geometry, Higgs fields, and the inverse problem}

\author{%
Steven Rayan \thanks{Department of Mathematics and Statistics, University of Saskatchewan, 106 Wiggins Road, Saskatoon, Saskatchewan, S7N\,5E6, Canada; rayan@math.usask.ca}
\and
  Thomas Stanley\thanks{Department of Mathematics and Statistics, University of Saskatchewan, 106 Wiggins Road, Saskatoon, Saskatchewan, S7N\,5E6, Canada; ths808@mail.usask.ca}
  \and
  Jacek Szmigielski\thanks{Department of Mathematics and Statistics, University of Saskatchewan, 106 Wiggins Road, Saskatoon, Saskatchewan, S7N\,5E6, Canada; szmigiel@math.usask.ca}
}

\date{September 12, 2018}

\makeatletter
\hypersetup{%
  pdfauthor={Steven Rayan, Thomas Stanley, Jacek Szmigielski},
  pdftitle={\@title},
}
\makeatother

\begin{document}

\maketitle

\begin{center}\emph{Dedicated to Emma Previato on the occasion of her 65th birthday.}\end{center}

\begin{abstract}
  We review the Calogero-Fran\c coise integrable system, which is a generalization of the Camassa-Holm system.  We express solutions as (twisted) Higgs bundles, in the sense of Hitchin, over the projective line.  We use this point of view to (a) establish a general answer to the question of linearization of isospectral flow and (b) demonstrate, in the case of two particles, the dynamical meaning of the theta divisor of the spectral curve in terms of mechanical collisions.  Lastly, we outline the solution to the inverse problem for CF flows using Stieltjes' continued fractions. 
  \end{abstract}

\tableofcontents{}

\section{Introduction}
\label{sec:intro}
The idea of viewing certain non-linear problems as arising from {\sl isospectral} deformations of linear operators goes back to P. D. Lax, who, in \cite{lax:kdv}, connected
the existence of infinitely-many integrals of motion for the Korteweg-de Vries equation (KdV) with an isospectral deformation of a linear operator $L_{u}$ parametrized by a solution $u$ to the KdV 
equation.  More concretely, introducing the one-dimensional (in $x$) Schr\"odinger  operator $L_{u}=-D^2 +u(x,t)$ he observed that the KdV equation $u_t+uu_x+u_{xxx}=0$ 
was equivalent to the operator equation 
\begin{equation}\label{eq:kdvLax} 
\dot L_{u}=[P,L_{u}], 
\end{equation} 
where $P$ is certain third-order differential operator depending on $u$ and 
its $x$ derivative.    This line of research was taken up by J. Moser, especially in the context of Hamiltonian systems with finitely many degrees of freedom
\cite{moser:three-integrable,moser:quadrics}.  
Some of these systems, like the finite Toda lattice or an $n$-dimensional rigid body \cite{manakov:rigid}, had a finite dimensional 
Lax pair and the dynamical problem despite being isospectral from 
the outset had no a priori relation to complex geometry even though 
a deeper analysis in each case was unequivocally pointing to the existence of such a connection.  This connection was established within a Lie-theoretic context 
in \cite{adler-van:Kac-Moody} leading to many years of fruitful 
interaction between Lie theory (mostly Kac-Moody Lie algebras) and 
the theory of integrable systems.  One of the decisive contributions to this 
theme was E. Previato's paper with M. Adams and J. Harnad \cite{adams:1988isospectral}, in some sense complementing the work of J. Moser \cite{moser:quadrics}.

In the mid-1970s, yet another class of integrable, finite-dimensional systems  was obtained from reductions 
of Lax integrable PDEs, i.e. famous finite-zone potentials \cite{novikovSP:periodic, dubrovin-novikov, lax:kdv-periodic, mckean:Hill, its-matveev:finitegap}, 
and led to the appearance of invariant \textit{spectral curves}.  
This had become a dominant research direction for many years to come and 
Emma beautifully reviewed this vast area in her 1993 lecture notes \cite{previato:1996seventy} placing emphasis on the 
old paper of Bourchnal and Chaundy \cite{burchnall:commuting}.  

The present paper is about a different occurrence of spectral curves, also due to the 
reduction from a PDE  given by a Lax pair equation, but the reduction 
is in smoothness.  We now turn to describing schematically the situation, leaving the details to \autoref{sec:CFgeneral}.  
The Camassa-Holm equation \cite{camassa-holm} (CH)
\begin{equation*} 
m_t+2mu_x+um_x=0, \quad m=u-u_{xx}.  
\end{equation*}
was invented as a model for nonlinear water waves with nonlinear dispersion. It has a Lax pair 
 \begin{equation*}
 L=-D^2 +\frac14-\lambda m, \qquad P=(\frac{1}{2\lambda}+u)D+\frac{u_x}{2},  
 \end{equation*} 
 from which it is clear that the central object in this endeavour is 
 $m$, while $u$ should be thought of as a potential producing $m$.  
 Even though the Lax equation has to be slightly modified (see \autoref{sec:CF}), 
 the computation is elementary for smooth $m$.  However, this is not so 
 if $m$ is non-smooth, for example if $m$ is a discrete measure, because then the Lax equation involves a multiplication of 
 distributions with overlapping singular supports and this leads to 
 certain subtle phenomena (see i.e. \cite{chang-szmigielski:mCH}).

It is the presence of spectral curves, which in this setting arise out of the reduction from a smooth $m$ to a 
 discrete measure $m$, that brings algebraic geometry into play.  To bring to bear this aspect fully, we recall that the work of Adams-Harnard-Previato \cite{adams:1988isospectral} is part of a sequence of results in the 1980s and early 1990s that translate classical integrable systems theory into the framework of complex algebraic geometry.  At the centre of this theme is the \emph{Hitchin system}, discovered in \cite{NJH:87} as an algebraically completely integrable Hamiltonian system defined on an enlargement of the cotangent bundle of the moduli space of stable holomorphic bundles on a fixed Riemann surface $X$ of genus $g\geq2$.  The entire system has a modular interpretation, namely as the moduli space of stable ``Higgs bundles'' on $X$, which consist of holomorphic vector bundles together with $1$-form-valued maps called ``Higgs fields''.  Higgs bundles themselves arise as solutions to a dimensional reduction of the self-dual Yang-Mills equations in four dimensions, as in \cite{NJH:86}.  The Hamiltonians for this system have a wonderfully explicit description in terms of characteristic data of the Higgs field.

Versions of the Hitchin system arise in lower genus, too. To accommodate surfaces $X$ with $g=0$ or $g=1$, one can make one of two (related) modifications.  On the one hand, the Higgs field can be allowed to take values in a line bundle other than the bundle of $1$-forms, leading to an integrable system studied in genus $0$ by P. Griffiths \cite{griffiths:lax} and A. Beauville \cite{AB:90} and in arbitrary genus by E. Markman \cite{EM:94}.  Retroactively, we refer to this as a \emph{twisted Hitchin system}, as it is formally a moduli space of Higgs bundles but with Higgs fields that have been twisted to take values in a line bundle of one's choosing.  The Hamiltonians have the same description as in the original Hitchin system but, generally speaking, the resulting integrable system is \emph{superintegrable}: it contains more Poisson-commuting Hamiltonians than are necessary.  In more geometric terms, the total space of the system is a torus fibration in which the base typically has dimension larger than that of the fibre.  The other modification is to puncture $X$ at finitely-many points and to allow the Higgs field to develop poles at these points, as in \cite{TSCH:90,BY:96,BB:04}.  Typically, one asks for the residues of the Higgs field at the poles to satisfy a certain Lie-theoretic condition, such as being semisimple.  This scheme has the virtue of preserving certain desirable properties of the original Hitchin system, such as the existence of a holomorphic symplectic form.  (In contrast, the twisted Hitchin systems generally fail to be globally symplectic and possess a family of degenerate Poisson structures that depend on a choice of divisor, as in \cite{EM:94}.)

A folklore belief is that every completely integrable system should be realizable as a Hitchin system of some kind, for some choice of Riemann surface $X$.  If true, this has the advantage of providing a systematic origin for spectral curves, namely as branched covers of the Riemann surface $X$.  A natural question is: when and how can a particular integrable system be identified with a Hitchin system?  In some sense, the original Hitchin systems for $g\geq2$ give rise to integrable systems of KdV / KP-type (for instance, \cite{HM:10}). Classically-known integrable systems tend to feature integrability in terms of elliptic integrals and hence involve the projective line $\mathbb P^1$ and elliptic curves.  For example, geodesic flow on the ellipsoid and Nahm's equations are twisted Hitchin systems on $X=\mathbb P^1$, as described in \cite{hitchin-segal-ward:is}.  Here, the Lax pair integrability can be expressed explicitly in terms of Laurent series in an affine chart on the $\mathbb P^1$.

In this article, we ask this question for the Calogero-Fran\c coise integrable system, which arises as a generalization of the Camassa-Holm dynamics.  We demonstrate how one can fit the CF integrable system into a twisted Hitchin system on $X=\mathbb P^1$, with Higgs fields taking values in the line bundle $\mathcal O(d)$, whose transition function is $z^d$ in the local coordinate.  One nice feature of this identification is that the theta divisor in the Jacobian of the spectral curve can be interpreted as a dynamical collision locus.  We demonstrate this explicitly in the case $d=2$.  Along the theta divisor, we also see a transition to a Hitchin system with poles of order $1$, capturing the singular dynamics algebro-geometrically.  Finally, we examine the inverse problem for CF from the point of view of continued fractions, in the sense of Stieltjes.

We hope that the mix of integrable systems theory and complex algebraic geometry in this article reflects some of the spirit of E. Previato's groundbreaking work over the past several decades.\\

\phantomsection
\addcontentsline{toc}{section}{Acknowledgements}
\subsection*{Acknowledgements} We are grateful to P. Boalch for useful discussions concerning Hitchin systems and Lax integrability.  The first and third named authors acknowledge the support of Discovery Grants from the Natural Sciences and Engineering Research Council of Canada (NSERC).  The second named author was supported by the NSERC USRA program.

\section{Calogero-Fran\c coise Hamiltonian System} \label{sec:CFgeneral}
The main reference for this section is \cite{beals-sattinger-szmigielski:CF}.  
We nevertheless present the main aspects of the setup to introduce notation and the main dynamical objects.  
F. Calogero and J.-P. Fran\c coise introduced in \cite{calogero-francoise:CF} 
a family of completely integrable Hamiltonian systems with Hamiltonian 
\begin{equation} \label{eq:CFHam}
H(x_1,\cdots, x_d, m_1,\cdots, m_d))=\tfrac12 \sum_{j,k=1}^d m_j m_k G_{\nu, \beta}(x_j-x_k)
\end{equation} 
where 
\begin{equation} \label{eq:Green}
G_{\nu,\beta}(x)=\frac{\beta_-}{2\nu}e^{-2\nu\abs{ x}}+\frac{\beta_+ }{2\nu}e^{2\nu\abs{x}}, 
\end{equation} 
and $\{x_1,\cdots, x_d\}$ and $\{m_1, \cdots, m_d\}$ are canonical positions and momenta respectively.  For future use we will define $\beta$ as a 
$2\times 2$ diagonal matrix $\textrm{diag}(\beta_-, \beta_+)$.  

The authors of \cite{calogero-francoise:CF} constructed explicitly $d$  Hamiltonians $\{H_j, j=1, \cdots, d\}$ and directly showed that they were in involution, i.e. $\{H_j, H_k\}=0$, with respect to the canonical Poisson bracket.  The special case $\beta_+=0, \beta_-=1$ 
was used as a motivating example and we briefly describe now this special case.  In 1993, R. Camassa and D. Holm \cite{camassa-holm} proposed what would turn out to be one of the most studied nonlinear partial differential equations of 
the last three decades, namely 
\begin{equation} \label{eq:CH}
m_t+2mu_x+um_x=0, \quad m=u-u_{xx}.  
\end{equation}
The equation was originally derived from the Hamiltonian for Euler's equation 
in the shallow water approximation.  One of the outstanding properties of 
the resulting equation is that it captures some aspects of "slope-steepening" and the breakdown of regularity of solutions, while at the same time it exhibits numerous intriguing 
aspects of Lax integrability, the connections to continued fractions of 
Stieltjes' type being one.  One feature that stands out in the present context is the existence of non-smooth solitons, dubbed {\sl peakons}.  These 
are obtained from the {\sl peakon ansatz}
\begin{equation} \label{eq:peakonAnstatz}
u(x,t)=\sum_{j=1}^d m_j(t)\, e^{-\abs{x-x_j(t)}}
\end{equation} 
for which $m$ becomes a finite sum of weighted Dirac measures 
\begin{equation} \label{eq:mpeakon}
m=2 \sum_{j=1}^d m_j(t) \delta_{x_j(t)}, 
\end{equation} 
and subsequently, upon substituting into \eqref{eq:CH}, one 
ends up with the systems of ODEs for positions $x_j$ and momenta $m_j$
\begin{equation}\label{eq:peakonODEs}
\dot x_j=u(x_j),  \qquad \dot m_j=-m_j\avg{u_x}(x_j), 
\end{equation} 
where $\avg{f}(x_j)$ denotes the arithmetic average of the right and left 
limits of $f$ at $x_j$.  Moreover, the peakon equation \eqref{eq:peakonODEs} 
is Hamiltonian with respect to the canonical Poisson bracket and 
Hamiltonian 
\begin{equation} \label{eq:HamCHpeakons} 
H(x_1,\dots,x_d,m_1,\dots,m_d)
  = \frac12 \sum_{i,j=1}^d m_i m_j e^{-\abs{x_i-x_j}}. 
  \end{equation} 
 The CH equation \eqref{eq:CH} has a Lax pair 
 \begin{equation}\label{eq:CHLaxpair}
 L=-D^2 +\frac14-\lambda m, \qquad P=(\frac{1}{2\lambda}+u)D+\frac{u_x}{2} 
 \end{equation} 
 whose compatibility indeed yields \eqref{eq:CH} (see, however, the discussion in \autoref{sec:CF}).   
 
It was shown in \cite{beals-sattinger-szmigielski:stieltjes, beals-sattinger-szmigielski:moment} that peakon equations can be explicitly integrated using classical results of analysis including the Stieltjes' continued fractions and the moment problem.  The CH peakon Hamiltonian 
 \eqref{eq:HamCHpeakons} was the starting point for the analysis in \cite{calogero-francoise:CF} and clearly the Hamiltonian  \eqref{eq:CFHam}
 is a natural generalization of \eqref{eq:HamCHpeakons}.  
 The fact that this generalization fits in with the CH equation \eqref{eq:CH} 
 was proven in \cite{beals-sattinger-szmigielski:CF}.  We will review the analysis based on that paper with due attention to the emergence of 
 a spectral curve and associated Riemann surface, both of which were absent from the analysis in \cite{calogero-francoise:CF} and were only in the background in \cite{beals-sattinger-szmigielski:CF}.

\section{CF flows: the Peakon Side}\label{sec:CF}
Given a measure $m\in \mathcal{M}(\R)$ and $\lambda \in \C$ we form the operator pencil 
\begin{equation} \label{eq:Llambda} 
L(\lambda)=D^2-\nu^2 -2 \nu \lambda  m.   
 \end{equation} 
We introduce another operator 
\begin{equation}
B(\lambda)=\big(\frac{1}{2\nu \lambda}-u\big)D+\tfrac12 u_x, 
\end{equation} 
and observe that 
\begin{equation}
[B(\lambda),L(\lambda)]  =2\nu\lambda (2u_xm +mu_x)-(m+\tfrac12 u_{xx}-2\nu^2 u)_x +\text{ mod $L(\lambda)$}, 
\end{equation}
where $\text{ mod $L(\lambda)$}$ means that this part vanishes on the 
kernel of $L(\lambda)$, from which 
we conclude that the operator equation (valid identically in $\lambda$ ) 
\begin{equation} \label{eq:gLax}
\dot L(\lambda)=[B(\lambda), L(\lambda)] +\text{ mod $L(\lambda)$}
 \end{equation} 
 implies 
 \begin{equation} \label{eq:CHF}
 m_t=(mD+Dm)u, \qquad [D(4\nu^2 -D^2)] u=2m_x. 
 \end{equation}
 \begin{remark} \label{rem:Manakov}
 In his 1976 paper S. Manakov \cite{manakov:triples}
 introduced a generalization of the Lax formalism.  The main 
 new aspect amounted to replacing 
 the standard Lax pair formulation $\dot L=[B,L]$ with 
 what would become known as a {\sl Manakov triple formulation} 
 by postulating the existence of another operator, say, $C$ such 
 that a generalized Lax equation $\dot L=[B,L]+CL$ holds.  Needless to say 
 the CH equation and in fact many other integrable equations have since been identified as satisfying some form of the Manakov triple formalism. 
 In the CH case, $C=-2u_x$.  

 \end{remark} 
 
 The main reason for reviewing this 
 derivation, despite its obvious affinity with the CH Lax pair \eqref{eq:CHLaxpair}, is to emphasize the second equation determining how 
 the potential $u$ is related to the measure $m$.  Upon one integration 
 we get $(4\nu^2 -D^2)u=2m+c$ for some $c$, resulting  in 
 shifting $u\rightarrow u+\frac{c}{4\nu^2}$ which can easily be absorbed by the Gallilean transformation $t=t', x=x'+ct$.  Thus we could assume that $c=0$ and that's precisely what was done in \cite{beals-sattinger-szmigielski:CF}.  
 Yet, in this paper we will choose a particular constant $c$, specified later, to fit more naturally with other developments.  
 We note that $\tfrac12 G_{\nu, \beta}$ in \eqref{eq:Green} is the most general even fundamental solution of $(4\nu^2 -D^2)$, provided $\beta_--\beta_+=1$.   Moreover, any rescaling of $G_{\nu,\beta}$ results in a rescaling of $H(x_1,\dots,x_d,m_1,\dots,m_d)$ which, in turn, can be 
 compensated by changing the time scale.  So the assumption $\beta_--\beta_+=1$ causes no loss of generality and will be in force for the remainder of the paper.
 
 We will concentrate from this point onward on the peakon sector whose definition we record to fix notation
 \begin{equation} \label{eq:um}
 u(x)=\sum_{j=1}^d m_j G_{\nu,\beta}(x-x_j)-C, \quad \qquad m=\sum_{j=1}^d m_j
 \delta_{x_j}. 
 \end{equation} 
Moreover, we assign the labels to positions in an increasing order $x_1<x_2<\cdots<x_d$.  In the CH peakon case ($\beta_+=0$), 
 the positivity of masses (momenta) $m_j$ is crucial for 
 the global existence of solutions \cite{beals-sattinger-szmigielski:moment} 
 so we make the same assumption that $m_j>0$ 
 until further notice.  We note that the evolution equation 
 \eqref{eq:CHF} has to be interpreted in the sense of distributions.  
 In particular 
 \begin{equation*} 
 m_t=\sum_{j=1}^d \big(\dot m_j \delta_{x_j}- \dot x_j m_j \delta_{x_j}^{(1)}\big)
 \end{equation*} 
 while the term $(mD+Dm)u$ has to be properly defined since the 
 singular supports of $m$ and $u$ coincide.  The regularization consistent 
 with Lax integrability turns out to be to assign the average value to $u_x$ at any 
 point $x_j$.  Thus $u_x \delta _{x_j}\stackrel{def}{=}\avg{u_x}(x_j) \delta_{x_j}$.  
 This point is explained in \cite{beals-sattinger-szmigielski:moment}.  
 The resulting CF Hamiltionian system has the same form as 
 \eqref{eq:peakonODEs} except for the definition of $u$ which is now 
 given by equation \eqref{eq:um}.  
 
 The presence of masses at $x_j$ 
 divides $\R$ into intervals $I_j=(x_{j-1}, x_j)$ with the proviso that $x_0=-\infty, \, x_{d+1}=+\infty$.  We will need the asymptotic behaviour of 
 $u$ on $I_1, I_{d+1}$ which follows trivially from the definition of $u$.  
 \begin{lemma} \label{lem:ump}
 \mbox{}\\*
 Let us set $M_\pm=\sum_{j=1}^d m_j e^{\pm2\nu x_j}$.  
 \mbox{}\\ 
 Then the asymptotic behaviour of $u$ is given by 
 \begin{align} 
 u(x)&=\frac{\beta_-}{2\nu}  M_-e^{2\nu x}+\frac{\beta_+}{2\nu}M_+e^{-2\nu x}-C, \qquad \qquad x\in I_1, \label{eq:uminus}\\
 \notag\\
 u(x)&=\frac{\beta_-}{2\nu} M_+ e^{-2\nu x}+\frac{\beta_+}{2\nu}M_-e^{2\nu x}-C, \qquad \qquad x\in I_{d+1}. \label{eq:uplus}
 \end{align}
 \end{lemma} 
 \begin{remark}
 The main difference between the CF scenario and the original peakon CH 
 case is the presence of a non-vanishing, actually 
 exponentially growing,  tail at $\pm \infty$.  
 Eventually, this has a real impact on the type of algebraic curve to which the the problem 
 is associated.  
 \end{remark}

 \subsection{Forward Problem} 
 We concentrate now on solving $L(\lambda)\Phi=0$ for $L(\lambda), m$ given by \eqref{eq:Llambda}, \eqref{eq:um} respectively.  The mathematics 
 involved is elementary, but one gets an interesting insight into the emergence of an underlying finite dimensional dynamical system.  
 The singular support of $m$ consists of positions $x_1<x_2<\cdots<x_d$.  
 On the complement, that is on the union of open intervals $\bigcup_{j=1}^{d+1} I_j$, we are solving 
 \begin{equation*}
 (D^2-\nu^2)\Phi=0.  
 \end{equation*} 
 Let us denote by $\Phi_j$ the restriction of $\Phi$ to $I_j$.  
 Then on each $I_j$, we have 
 \begin{equation*} 
 \Phi_j=a_je^{\nu x}+b_j e^{-\nu x},  \qquad 1\leq j\leq d+1;
 \end{equation*} 
 however, while crossing the right endpoint $x_j$, we have 
 \begin{equation*} 
 \Phi_j(x_{j})=\Phi_{j+1}(x_{j}), \qquad [D\Phi](x_{j})=(D\Phi_{j+1}-D\Phi_j)(x_{j})=2\nu \lambda \Phi_{j}(x_{j}).  
 \end{equation*} 
 If we use as a basis $\{e^{\nu x}, e^{-\nu x}\}$, then 
 $\Phi_j$ can be identified with $\begin{bmatrix} a_j\\ b_j \end{bmatrix}$ 
 and the last equation can be written 
 \begin{equation} \label{eq:transitionj} 
 \Phi_{j+1}=T_j \Phi_j, \qquad \qquad T_j=I+\lambda m_j\begin{bmatrix}1&e^{-2\nu x_j}\\ -e^{2\nu x_j}& -1 \end{bmatrix}, 
 \end{equation} 
 where $I$ denotes the $2\times 2$ identity matrix.

 We now define the transition matrix 
 \begin{equation} \label{eq:transition}
 T=T_d\dots T_1, 
 \end{equation} 
 which maps $\Phi_1\rightarrow \Phi_{d+1}$ .  In the next step we want to 
  determine the time evolution of $T$.  
 \begin{lemma} \label{lem:preLax}
 Let 
 \begin{align*} 
 B_-&=\begin{bmatrix} \frac{1}{2\lambda}+C& \beta_-M_-\\\\-\beta_+M_+ & -\frac{1}{2\lambda}-C \end{bmatrix},  \\\\
 B_+&=\begin{bmatrix} \frac{1}{2\lambda}+C& \beta_+M_-\\\\-\beta_-M_+ & -\frac{1}{2\lambda} -C\end{bmatrix}.
\end{align*} 
Then 
\begin{equation}
\dot T=B_+T-TB_-. 
\end{equation} 
 
 \end{lemma}

 \begin{proof} 
 The proof can be found in \cite{beals-sattinger-szmigielski:CF} but for the sake of completeness and to give the reader a sense of the origin of the geometric 
 underpinnings of the CF system we present an economical version of the argument.  
 We start by observing that the generalized Lax \autoref{eq:gLax} 
 implies 
 that 
 \begin{equation}\label{eq:ZCC}
L(\lambda)(\frac{\partial}{\partial t}-B(\lambda))\Phi=(\frac{\partial}{\partial t}-B(\lambda))L(\lambda) \Phi=0.  
\end{equation}
 Suppose we denote by $\hat \Phi$ the fundamental solution 
 normalized to be $\hat \Phi_1=I$ in $I_1$ (corresponding to $e^{\nu x}$ and $e^{-\nu x}$ as linearly independent solutions).  Then \eqref{eq:ZCC} implies that 
 \begin{equation}
(\frac{\partial}{\partial t}-B(\lambda))\hat \Phi=\hat \Phi C(t)
\end{equation}
where $C(t)$ is a $2\times 2$ matrix whose entries do not depend on $x$.  
Let us denote by $B_j$ the restriction of $B(\lambda)$ to $I_j$.   Then $\hat \Phi_{j}=T_{j-1}\cdots T_1\hat \Phi_1$, and 
$$(\frac{\partial}{\partial t}-B_j)\hat \Phi_j=\hat \Phi_j C(t)=T_{j-1}\cdots T_1\hat \Phi_1C(t)=T_{j-1}\cdots T_1(\frac{\partial}{\partial t}-B_1)\Phi_1$$
Hence 
$(\frac{\partial}{\partial t}-B_j)\hat \Phi_j=T_{j-1}\cdots T_1(\frac{\partial}{\partial t}-B_1)\hat \Phi_1$.  In particular, for $j=d+1$ and denoting 
$B_{d+1}=B_+$ and $B_1=B_-$, we obtain 
\begin{equation}
(\frac{\partial}{\partial t}-B_+)T=-TB_1, 
\end{equation}
since $\hat \Phi_1=I$.  Finally, the computation of the matrix of $B_-$ and 
$B_+$ in the basis $\{e^{\nu x}, e^{-\nu x}\}$ follows readily from \autoref{lem:ump}. 
\end{proof}

We observe that $B_-$ and $B_+$ are very closely related: they differ only 
by the placement of $\beta_-$ and $\beta_+$, which are interchanged by conjugating with the diagonal matrix $\beta$, namely, 
\begin{equation} 
B_-\beta=\beta B_+
\end{equation} 
This leads us to an interesting corollary.  
\begin{corollary} \label{cor:trueLax}
Let $A(z)=z^d T(\lambda=\frac{1}{z})\beta$ and set $B(z)=B_+(\lambda=\frac{1}{z})$.  Then $A(z)$ satisfies the Lax equation
\begin{equation}\label{eq:trueLax}
\dot A(z)=[B(z),A(z)].  
\end{equation}
\end{corollary} 
We observe that now 
both $A(z)$ and $B(z)$ are matrix valued polynomials 
in $z$ of degrees $d$ and $1$ respectively.  Clearly, we 
can associate to $A(z)$  a  {\sl spectral curve} \begin{center} 
\boxed{
C(w,z)=\det{\big(w I -A(z)\big)}=0,} 
\end{center} 
or more succinctly 
\begin{equation} 
C(w,z)=\{(w,z): 
w^2=\textrm{tr} (A(z)) w-z^{2d}\det{\beta}\}. 
\end{equation}
We remark that the compactification of the the affine curve $C(w,z)=0$, which we will denote by $Y$, is a $2$-fold cover of $\mathbb{P}^1$ on which $w$ is single valued.

\subsection{Higgs Fields and Properties of $A$ and $B$}
The first observation about $A$ is that it differs only 
by a multiplication by $z^d \beta$ from $T(\lambda=1/z)$ in \autoref{eq:transition}.  
Thus $A$ is a matrix valued polynomial of degree $d$ in $z$, 
$A=\sum_{j=0}^d A_j z^j $.  By factoring $z^d$ this 
defines a holomorphic function around $\infty$ with a local parameter 
$\tilde z =\frac{1}{z}$.

As a matrix, each $A(z)$ acts by left multiplication of course on the vector space $V=\C^2$.  The passage from linear algebra to geometry occurs by viewing $z\in\mathbb P^1$ as parametrizing a family of such vector spaces, and so the whole object $A$ acts by left multiplication on the vector bundle $E=\mathbb P^1\times\C^2$.  This is a rank-$2$ vector bundle with trivial holomorphic structure.  Typically, the isomorphism class of holomorphic line bundles on $\mathbb P^1$ with transition function $z^n$ is denoted by $\mathcal O(n)$.  The first Chern class of $\mathcal O(n)$, identified with an integer via Poincar\'e duality, is $n$.  This integer is called the \emph{degree} of the line bundle and we write $\deg\mathcal O(n)=n$.  The degree tells us the number of times that a generic holomorphic section vanishes.  The now-classical splitting theorem of Birkhoff and Grothendieck says that every holomorphic bundle on $\mathbb P^1$ decomposes uniquely (up to ordering) as a sum of holomorphic line bundles.  In the case of our $E$, this is $E=\mathcal O\oplus\mathcal O$, where $\mathcal O:=\mathcal O(0)\cong\mathbb{P}^1\times\C$ is the trivial line bundle (or ``structure sheaf'') on $\mathbb P^1$.  The degree is additive with respect to both tensor products and direct sums of line bundles.  Hence, $\deg E=0$.\\

In this framework, $A$ is a holomorphic map from $E$ to itself, tensored by $\mathcal O(d)$:\begin{center} \boxed{
A \in H^0(\mathbb{P}^1,\mbox{End}(E)\otimes\mathcal{O}(d))}
\end{center}

The data of $A$ is the \emph{Higgs field} in our set-up; together, $(E,A)$ is the \emph{Higgs bundle}.  We refer to the dimension of the fibre of $E$ as the \emph{rank} of $E$.  In this case, the rank of $E$ is $r=2$.  We also ought to remark that the pair $(E,\phi)$ is a ``twisted'' Higgs bundle relative to the formulation of Hitchin \cite{NJH:86}, as the Higgs bundles coming from gauge theory would be $\mathcal O(-2)$-valued on $\mathbb P^1$ while our $d$ is strictly positive.

We will now compute the coefficients $A_j$ for the Higgs field in terms of the 
original data $\{x_1, \cdots, x_d, m_1, \cdots, m_d\}$.  This step is not 
relevant to the geometry of the problem but it is crucial if one actually wants to solve the original peakon equation 
$\dot x_j =u(x_j), \quad \dot m_j =-m_j \avg{u_x}(x_j)$.  

Recall that 
\begin{equation*} 
m=\sum_{j=1}^d m_j \delta_{x_j}, \quad x_1<x_2<\cdots<x_d, 
\end{equation*}
and (see \eqref{eq:transitionj})  
\begin{equation*}
T_j=I+\lambda m_j\begin{bmatrix} 1 & e^{{ -}2\nu x_j} \\
-e^{2\nu x_j}& -1  \end{bmatrix}.  
\end{equation*}

The first elementary observation is that if we set 
 $X_j=\begin{bmatrix} 1&e^{-2\nu x_j}\\ -e^{2\nu x_j}& -1 \end{bmatrix}$ then 
\begin{equation}\label{eq:Xj}
X_j=\begin{bmatrix} 1\\-e^{2\nu x_j} \end{bmatrix} \begin{bmatrix} 
1& e^{-2\nu x_j} \end{bmatrix}. 
\end{equation} 
 
 \begin{definition}
  The binomial coefficient $\binom{S}{j}$
  denotes the collection of $j$-element subsets
  of the set~$S$,
  and $[d]$ denotes the integer interval $\{ 1,2,3, \dots, d \}$.
  Moreover, the elements of a set
  $I \in \binom{[d]}{j}$ are labeled in increasing order:
  $I = \{ i_1 < i_2 < \dots < i_j \}$.
  Finally, given a collection of matrices $\{a_1, a_2, \cdots, a_d\} $
  we denote the ordered product of matrices labeled by the multi-index 
  $I$ as $a_I\stackrel{def}{=} a_{i_j}a_{i_{j-1}}\cdots a_{i_1}$.  
 
\end{definition}
\begin{lemma} 
Given a multi-index $I \in \binom{[d]}{j}$ and the set of 
matrices $\{X_1, X_2, \cdots, X_d\}$ as in \autoref{eq:Xj} then 
\begin{equation} 
X_I=\Big(\prod_{k=1}^{j-1} \big( 1-e^{-2\nu(x_{i_{k+1}}-x_{i_k})}\big) 
\Big) \, \begin{bmatrix} 1 & e^{-2\nu x_{i_1}}\\-e^{2\nu x_{i_j}}& -e^{2\nu(x_{i_j}-x_{i_1})}\end{bmatrix} 
\end{equation} 
with the proviso that the empty product is taken to be $1$ when $j=1$.  
\end{lemma} 
\begin{proof} 
It suffices to observe that 
\begin{equation*}
\begin{split} X_jX_i=
\begin{bmatrix} 1\\-e^{2\nu x_j} \end{bmatrix} \begin{bmatrix} 
1& e^{-2\nu x_j} \end{bmatrix} \begin{bmatrix} 1\\-e^{2\nu x_i} \end{bmatrix} \begin{bmatrix} 
1& e^{-2\nu x_i} \end{bmatrix}=\\(1-e^{-2\nu (x_j-x_i)})\begin{bmatrix} 1\\-e^{2\nu x_j} \end{bmatrix} \begin{bmatrix} 
1& e^{-2\nu x_i} \end{bmatrix}, \qquad i<j, 
\end{split} 
\end{equation*} 
and then proceed by induction on the number of terms.  
\end{proof}

\vspace{1cm} 

Given a multi-index $I \in \binom{[d]}{j}$ we denote 
\begin{equation} \label{eq:fi}
\prod_{k=1}^{j-1} \big( 1-e^{-2\nu(x_{i_{k+1}}-x_{i_k})}\big)\stackrel{def}{=}f_I. 
\end{equation}

\begin{theorem} \label{thm:Astructure}
\begin{equation} 
T=I + 
\sum_{j=1}^d\lambda^{j} \sum_{
I \in \binom{[d]}{j}}  f_I m_I \begin{bmatrix} 1 & e^{-2\nu x_{i_1}}\\-e^{2\nu x_{i_j}}& -e^{2\nu (x_{i_j}-x_{i_1})}\end{bmatrix} 
\end{equation}
In particular,
\begin{align} 
T&=\lambda^{d} f_{[d]}m_{[d]}\begin{bmatrix}
1& e^{-2\nu x_1}\\
-e^{2\nu x_d}&-e^{2\nu(x_d-x_1)} \end{bmatrix} +\textrm{O}(\lambda^{d-1}), \qquad & \lambda \rightarrow \infty,  \\
A(z)&=\beta z^d +\sum_{j=1}^d z^{d-j} \sum_{
I \in \binom{[d]}{j}}  f_I m_I \begin{bmatrix} \beta_- & \beta_+e^{-2x_{i_1}}\\-\beta_-e^{2x_{i_j}}& -\beta_+e^{2(x_{i_j}-x_{i_1})}\end{bmatrix}, \\
\textrm{tr}A(z)&= f_{[d]}m_{[d]}(\beta_--\beta_+e^{2\nu(x_d-x_1)})+ \textrm{O}(z), \qquad &z \rightarrow 0,  \\
A(z)&=\beta z^d  +\begin{bmatrix} \beta_- M&\beta_+ M_-\\ -\beta_-M_+&-\beta_+M \end{bmatrix}z^{d-1} + \textrm{O}(z^{d-2}), \qquad &z \rightarrow \infty, 
\end{align}
where  $M=\sum_{j=1}^d m_j$ is the total mass (momentum) and 
$M_{\pm}=\sum_{j=1}^d e^{\pm 2\nu x_j} m_j$ (see \autoref{lem:ump}).  
\end{theorem} 
 In view of the invariance of $\textrm{tr}A(z)$ we immediately have: 
\begin{corollary} \label{cor:conservation}
Under the CF flow 
\begin{enumerate} 
\item the total mass (momentum) $M$ is conserved; 

\item if at $t=0$ all masses $m_j(0)$ have the same sign then  they can not collide, meaning, $x_i(t)\neq x_j(t), i\neq j$ for 
all times.  
\end{enumerate} 
\end{corollary} 
\begin{proof} 
We start off by noting that individual masses  cannot change signs.  
Indeed the equation of motion $\dot m_j=-m_j\avg{u_x}(x_j)$ implies 
that $m_j(t)$ has the same sign as $m_j(0)$.  
The first statement follows immediately from the invariance of $\textrm{tr}A(z)$, since $\textrm{tr} A(z)=(\beta_-+\beta_+)z^d +Mz^{d-1}+O(z^{d-2})$. For the second claim it   suffices to prove that the neighbours cannot collide.  
First, we observe that  $m_{[d]}$ will remain bounded away from 
zero if masses have the same sign.  Suppose now the constant of motion $f_{[d]}m_{[d]}(\beta_--\beta_+e^{2\nu(x_d-x_1)})$ is not zero at $t=0$. Then, $f_{[d]}=\prod_{k=1}^{d-1}\big( 1-e^{-2(x_{k+1}-x_{k})}\big)$ will remain non-zero 
if   
$\beta_--\beta_+e^{2\nu(x_d(0)-x_1(0))}\neq 0$. 
\end{proof}

\section{Linearization} \label{sec:linearization}

One of the advantages of the Higgs bundle framework is that the moduli space of Higgs bundles on a Riemann surface $X$ is fibred by tori, each of which is the Jacobian of a spectral curve for a Higgs field.  In a sense, all possible spectral curves covering the given $X$ appear in the moduli space.  At the same time, the total space of the moduli space is an algebraically completely integrable system \cite{NJH:87} that extends, in a canonical way, the phase space structure of the cotangent bundle to the moduli space of bundles on $X$.  Furthermore, explicit Hamiltonians are given by invariants of the Higgs fields --- these are the components of the so-called ``Hitchin map'', which sends a Higgs field $A$ to the coefficients of its characteristic polynomial.

The \emph{spectral correspondence} identifies isomorphism classes of Higgs bundles with a fixed characteristic polynomial with isomorphism classes of holomorphic line bundles on the associated spectral curve $Y$.  This correspondence is developed for $1$-form-valued Higgs bundles on $X$ of genus $g\geq2$ by Hitchin \cite{NJH:87} and for arbitrary genus and Higgs fields taking values in an arbitrary line bundle $L\to X$ by Beauville-Narasimhan-Ramanan \cite{BNR:89} --- see also \cite{EM:94,DM:96}.  The essence of the correspondence is that the eigenspaces of a Higgs field $A$ for a holomorphic vector bundle $E\to X$ form a line bundle $S$ on $Y$, which as a curve is embedded in the total space of $L$ (since $A$ is $L$-valued and so the eigenvalues are sections of $L$).  If $r$ is the rank of $E$, then $Y$ will be an $r$-sheeted cover of $X$.  If $\pi$ is the projection from the total space of $L$ to $X$, then the direct image $(\pi|_Y)_*S$ is a vector bundle $E'$ on $X$, isomorphic to the original vector bundle $E$.  Let $z$ be a local coordinate on $X$ (just as with $X=\mathbb P^1$ in the preceding discussion).  The map that multiplies sections $s(z)$ of $S$ by $w(z)$, where $w(z)$ is the corresponding point on $Y$, is the action of eigenvalues on eigenspaces.  This map pushes forward to an $L$-valued endomorphism $A'$ of $E$ whose spectrum is $Y$.  This operation that starts with $(E,A)$ and ends with the isomorphic Higgs bundle $(E',A')$ is, at almost every point of $X$, diagonalization.

When $X=\mathbb P^1$ and $L=\mathcal O(d)$, Lax partners $B$ that complement the Higgs fields $A$ can be systematically computed and written down explicitly, as carried out by Hitchin in \cite{hitchin-segal-ward:is}.  The only assumption necessary on $A$ is that it has a smooth, connected spectral curve $Y$, which is a fairly generic property.  By Bertini's theorem on pencils of divisors, the generic characteristic polynomial produces a smooth spectral curve.  At the same time, the moduli space consists only of ``stable'' Higgs bundles, which are Higgs bundles $(E,A)$ with a restriction on which subbundles of $E$ can be preserved by $A$. Normally, this is imposed to ensure that the moduli space is topologically well-formed.  The precise condition for Higgs bundles, called \emph{slope stability}, originates in \cite{NJH:87}, and in our case is as follows: $(E,A)$ on $\mathbb P^1$ is \emph{semistable} if for each nonzero, proper subbundle $U$ with $A(U)\subseteq U\otimes\mathcal O(d)$ we have that$$\frac{\deg(U)}{\mbox{rank}(U)}\leq\frac{\deg(E)}{\mbox{rank}(E)}.$$Generically, $A$ will be \emph{very stable}, meaning that it preserves no subbundle whatsoever (other than $E$ itself and $0$).  This corresponds to the spectral curve being connected.  

Embedding a known integrable system into a Hitchin-type system on $\mathbb P^1$ can therefore lead to new insights about linearized flows.  Having produced such an embedding for CF --- we established the existence of Higgs fields $A$ acting on $E=\mathcal O\oplus\mathcal O$ that solve the peakon equation --- we may now compute the linearization along the Hitchin fibres directly.

\begin{theorem} \label{thm:linCF} Let $A(z,t)\in H^0(\mathbb P^1,\mbox{End}(E)\otimes\mathcal O(d))$ be a solution to the peakon equation given above, such that its spectral curve $Y$ (constant for all $t$) is a double cover of $X=\mathbb P^1$ that is embedded as a non-singular, connected subvariety in the total space of $\mathcal O(d)$. 
Let $J^{\tilde g-1}$ be the Jacobian of degree $\tilde g-1$ line bundles on $Y$ and 
$\Theta$ be the theta divisor and $\tilde g$ is the genus of $Y$.  
The CF flow as given by \autoref{cor:trueLax} is linearized 
on $J^{\tilde g-1}\setminus \Theta$.  
\end{theorem} 
\begin{proof} 
The problem of linearization for matrix Lax equations with a spectral 
parameter $z$ has been studied by several authors \cite{hitchin-segal-ward:is, griffiths:lax}
We will use pertinent to this problem material from N. Hitchin's lectures 
in \cite{hitchin-segal-ward:is}.  For the Lax equation 
of the type $\dot A(z)=[B(z), A(z)]$ with  polynomial matrix valued 
$A(z), B(z)$ the linearization on $J^{g-1}\setminus \Theta$ happens 
if and only if $B(z)$ has a specific form dependent on $A(z)$ (\cite[Lecture 5]{hitchin-segal-ward:is}):
\begin{equation*}
B(z)=\gamma(z,A(z))^+,
\end{equation*}
where 
\begin{equation*}
\gamma(z,w)=\sum_{i=1}^{r-1}\frac{b_i(z) w^i}{z^N}.  
\end{equation*}
In this formula, $r$ is the degree (in $w$) of the spectral curve $\det(w-A(z))=0$, 
$N$ is an integer, $b_i(z)$ are polynomials in $z$, and $^+$ means 
the projection on the polynomial part.  In our case $r=2$ so 
$\gamma(z,w)=\frac{b(z) w}{z^N}$.  We now choose 
$b(z)=1$ and $N=d-1$ and settle with $B(z)=\big(\frac{A(z)}{z^{d-1}}\big)^+$.  
Then by \autoref{thm:Astructure} we get 
\begin{equation*}
B(z)=\begin{bmatrix} \beta_-z&0\\0&\beta_+z \end{bmatrix}+ \begin{bmatrix} \beta_-M&\beta_+M_- \\-\beta_- M_+&-\beta_+ M \end{bmatrix}.  
\end{equation*} 
However, $B(z)$ is not unique; any power of $A(z)$ can be added to $B(z)$ 
without changing the Lax equation, in particular any multiple of 
the identity can be added with impunity.  For example, 
by writing $\beta_-=\tfrac12 (\beta_-+\beta_+)+\tfrac12(\beta_--\beta_+)$ 
and the same for $\beta_+$, we see that one can take 
\begin{equation*} 
B(z)=\begin{bmatrix} \frac z2&0\\0&-\frac z2 \end{bmatrix}+ \begin{bmatrix} \frac{\beta_-+\beta_+}{2}M&\beta_+M_- \\-\beta_- M_+&-\frac{\beta_-+\beta_+}{2} M \end{bmatrix}.
\end{equation*} 
It suffices now to set 
\begin{center} \boxed{
C=\frac{\beta_-+\beta_+}{2} M }
\end{center} 
in \autoref{lem:preLax}
to complete the proof.

\end{proof}

\section{Two CF Peakons: Collisions}
  In this section we will analyze, as a concrete example, the case of $d=2$ to get a better insight 
  into the global existence of 
  CF flows, but also to demonstrate the special dynamical meaning of the theta divisor $\Theta\subset J^{\tilde g-1}$.

We recall that the moduli space of stable twisted Higgs bundles with $d=2$ on $\mathbb P^1$ on a rank $2$ holomorphic bundle of degree $0$ was studied algebro-geometrically by the first named author in \cite{SR:13}.  Over $\mathbb C$, there is a $9$-dimensional moduli space of such Higgs bundles.  The base of the Hitchin fibration is $8$-dimensional, reflecting the extreme underdetermined nature of the integrable system: we only need $2$ real Hamiltonians, as per the dimension of the fibre, but we actually have a $16$-dimensional space of such Hamiltonians available to us.

Theorem 6.1 in \cite{SR:13} characterizes exactly which holomorphic bundles $E$ with these topological invariants admit the structure of a semistable Higgs bundle $(E,A)$.  These are precisely $E=\mathcal O\oplus\mathcal O$ and $E=\mathcal O(1)\oplus\mathcal O(-1)$.  (A rank $2$, degree $0$ vector bundle $E=\mathcal O(k)\oplus\mathcal O(-k)$ with $k\geq2$ will necessarily have a preserved sub-line bundle $\mathcal O(k)$ and $k>0$, which violates stability as posed above.)  It is also shown in Proposition 8.1 in the same reference that, for $d=2$, there is up to isomorphism a unique Higgs field $A$ for the bundle $E=\mathcal O(1)\oplus\mathcal O(-1)$, once the characteristic coefficients of $A$ have been fixed.  These Higgs fields form a section of the Hitchin system here, intersecting each torus in this unique point.  In the isospectral problem, where the characteristic coefficients are constant, we thus have a unique Higgs field $A$ corresponding to $E=\mathcal O(1)\oplus\mathcal O(-1)$ in the Hitchin fibre determined by $\beta_+,\beta_-,M,\nu$.  All of the remaining Higgs bundles in the fibre are ones with $E=\mathcal O\oplus\mathcal O$. 

This begs the question of the meaning of this unique point.  Algebro-geometrically, the meaning is clear: it is the theta divisor $\Theta$, or rather a twist of it.  To see this, note that a Higgs bundle $(E,A)$ valued in $\mathcal O(2)$ with rank $2$ and degree $0$ yields, via the spectral correspondence, a degree $2$ line bundle $S$ on the spectral curve $Y$, which is a genus $\tilde g=1$ curve that covers $\mathbb P^1$ $2:1$.  (See again Section 8 of \cite{SR:13}.)  In general, the genus of the spectral curve can be computed through the Riemann-Hurwitz formula by noting the number of zeros of the determinant of $A$, which is $2d$.  As such, we will have $\tilde g=d-1$ in general.  Subsequently, $\deg(S)$ can then be computed via the Grothendieck-Riemann-Roch Theorem.  This, in fact, will always be $d$ and so $\deg(S)=2$ in this case.  (It is perhaps useful to point out that the degree of $S$ need not match the degree of $E$, as ramification in general disturbs such invariants.  A precise formula is given in Proposition 4.3 of Chapter 2 of \cite{hitchin-segal-ward:is}.)  All in all, the spectral correspondence employs $J^2$ of $Y$ while the theta divisor is a subvariety of $J^0$, as $\tilde g-1=0$ in this case.

Note that $\Theta$ is complex codimension $1$ in $J^0$.  In this particular case, it is a single point: the isomorphism class of the trivial line bundle $\mathcal O_Y:=Y\times\mathbb C$.  The remaining points in $J^0$ are holomorphic line bundles of degree $0$ for which there is no global nonzero holomorphic section.   To bridge the gap between $J^0$ and $J^2$, note that when $E=\mathcal O\oplus\mathcal O$, we have $E\otimes\mathcal O(-1)\cong\mathcal O(-1)\oplus\mathcal O(-1)$, which has no nonzero holomorphic sections.  On the other hand, when $E=\mathcal O(1)\oplus\mathcal O(-1)$, we have $E\otimes\mathcal O(-1)\cong\mathcal O\oplus\mathcal O(-2)$, which has a $1$-dimensional space of holomorphic sections.  Upstairs on $Y$, the associated line bundle is $S\otimes\pi^*\mathcal O(-1)$, where $\pi$ is again the projection of $\mathcal O(2)$ (restricted to $Y$).  Since $S$ is a degree $2$ cover, the degree of $S$ is shifted by $2\cdot(-1)$ when we twist by $\pi^*\mathcal O(-1)$.  In other words, we get $E=\mathcal O(-1)\oplus\mathcal O(-1)$ as the pushforward of any line bundle in $J^0$ except for if we push forward $\mathcal O_Y$ (i.e. the theta divisor), which instead gives us $E=\mathcal O\oplus\mathcal O(-2)$.

Translating all of this back to $J^2$, every line bundle in $J^2$ pushes forward to give $E=\mathcal O\oplus\mathcal O$ except for the distinguished line bundle $\mathcal O_Y\otimes\pi^*\mathcal O(1)$, which pushes forward to give $E=\mathcal O(1)\oplus\mathcal O(-1)$ --- and, along with it, a unique Higgs field $A$ for this bundle whose spectrum is $Y$.

Dynamically, we wish to demonstrate that this unique Higgs bundle structure on $\mathcal O(1)\oplus\mathcal O(-1)$, and hence the theta divisor itself, is a collision solution, extending the smooth dynamics presented earlier.

Now, recall the transition matrix $T$ for $d=2$:
$$T = T_2 T_1 = \left(I + \lambda m_2 \begin{bmatrix} 1 & e^{-2 \nu x_2} \\ -e^{2 \nu x_2} & -1 \end{bmatrix}\right) \left(I + \lambda m_1 \begin{bmatrix} 1 & e^{-2 \nu x_1} \\ -e^{2 \nu x_1} & -1 \end{bmatrix}\right)=$$

$$ I + \lambda \begin{bmatrix} m_1 + m_2 & m_1 e^{-2\nu x_1} + m_2 e^{-2 \nu x_2} \\ -(m_1 e^{2 \nu x_1} + m_2 e^{2 \nu x_2}) & - (m_1 + m_2) \end{bmatrix} + \lambda^2 m_1 m_2 \begin{bmatrix} 1 - e^{2 \nu (x_1 - x_2)} & e^{-2 \nu x_1} - e^{-2\nu x_2} \\ e^{2 \nu x_1} - e^{2 \nu x_2} & 1 - e^{2\nu (x_2 - x_1)}\end{bmatrix}.$$

Also, recall that by \autoref{cor:trueLax} $A(z) = z^2 T(\frac 1 z) \beta$, and $M = m_1 + m_2$. Hence, 

$$A(z) ={\tiny \begin{bmatrix} \beta_- (z^2 + zM + m_1 m_2 (1 - e^{2 \nu (x_1 - x_2)}) & \beta_+(z (m_1 e^{-2 \nu x_1} + m_2 e^{-2 \nu x_2}) + m_1 m_2 (e^{-2 \nu x_1} - e^{-2 \nu x_2}) \\ \beta_-(-z(m_1 e^{2\nu x_1} + m_2 e^{2 \nu x_2}) + m_1 m_2 (e^{2 \nu x_1} - e^{ 2\nu x_2})) & \beta_+ (z^2 - zM + m_1 m_2 (1 - e^{2 \nu (x_2 - x_1)})) \end{bmatrix}}. $$
Since 
$\Tr A(z)$ must be invariant, i.e. since the expression
$$z^2 (\beta_+ + \beta_-) + zM + m_1 m_2 (1 - e^{2\nu (x_1 - x_2)})(\beta_- - \beta_+ e^{2 \nu (x_2 - x_1)})$$ is time independent, it must follow that both  $M$ and \begin{center} 
\boxed{C_2=m_1 m_2 (1 - e^{2\nu (x_1 - x_2)})(\beta_- - \beta_+ e^{2 \nu (x_2 - x_1)})} 
\end{center} 
are invariant. By the relation $\beta_- - \beta_+ = 1$, both $\beta_+$ and $\beta_-$ are also determined by $\Tr A(z)$.
We recall that by \autoref{cor:conservation} there are no 
collisions ($x_1=x_2$) if masses $m_1,m_2$ are of the same sign.  
However, when masses have opposite signs the collision will occur exactly as they do in the $\beta_+=0$ case \cite{camassa-holm, beals-sattinger-szmigielski:moment}.

Numerical solutions to the $d=2$ case show a collision between particles $x_1$ and $x_2$ with various initial conditions.
The positions until the collision occurs for initial conditions $x_1(0)=1,x_2(0)=2,m_1(0)=5,m_2(0)=-1,\nu = 2, \beta_+ = 0.018$ are shown by the following graph:
\begin{figure}[h]
\begin{center}
\includegraphics[width=0.8\linewidth]{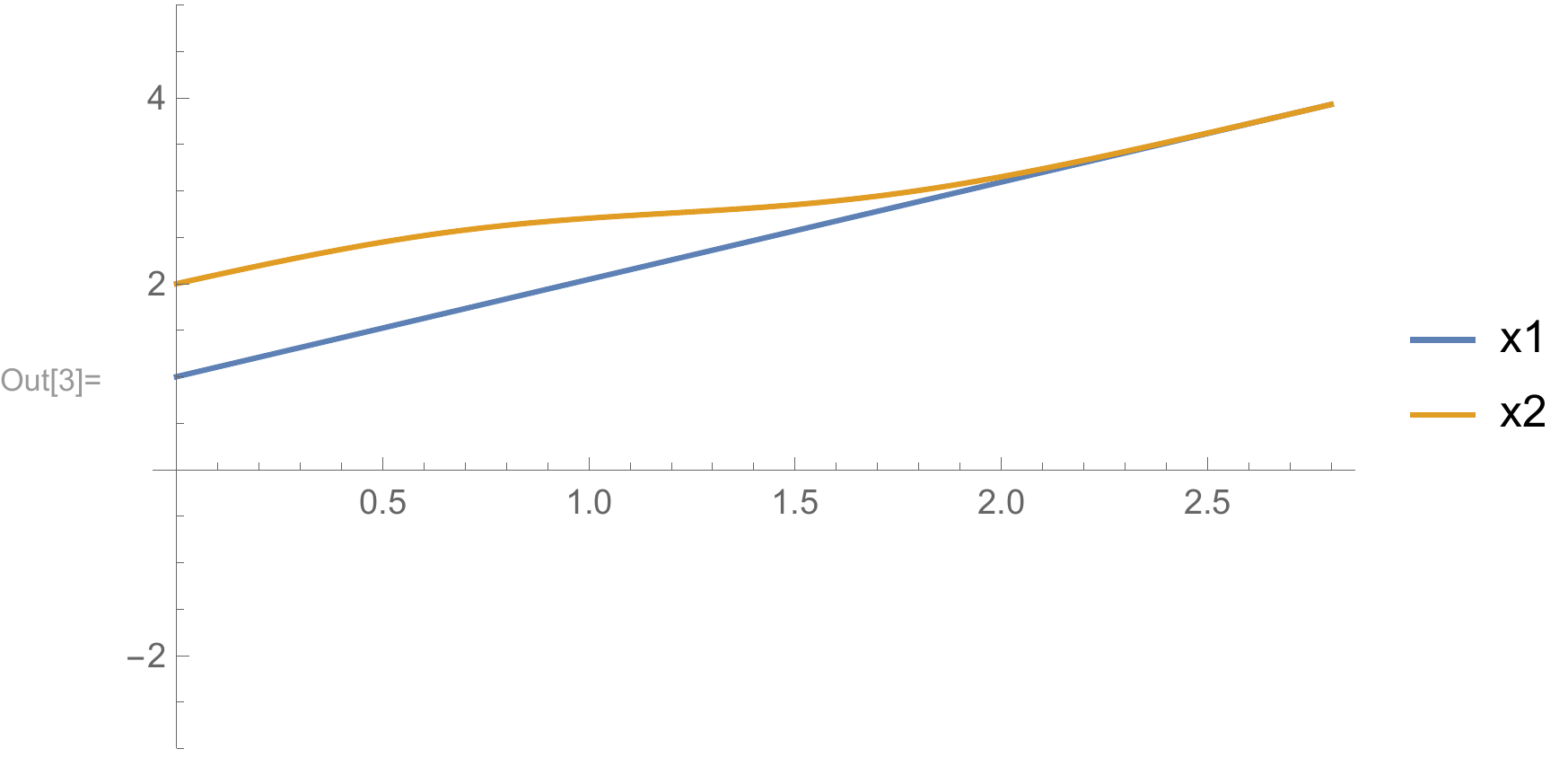}
\end{center}
\end{figure}

The masses also grow very large near the collision point with the same initial conditions:
\begin{figure}[h]
\begin{center}
\includegraphics[width=0.8\linewidth]{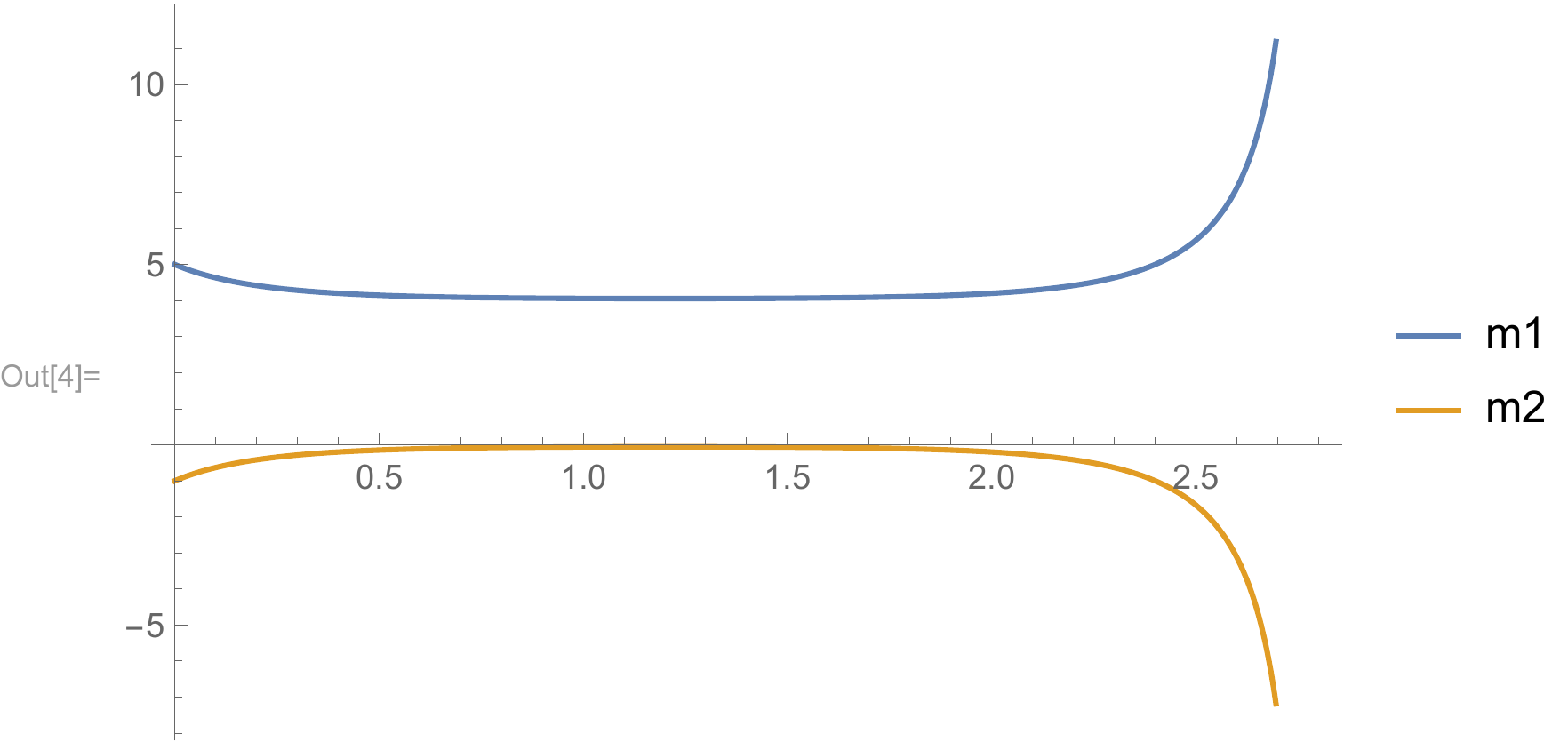}
\end{center}
\end{figure}

Let us briefly describe on a heuristic level 
the mechanics of such collisions.  Suppose the masses have opposite signs 
and $C_2<0$.  When $x_2-x_1$ becomes small the masses (momenta) 
grow large, one becoming large negative, the other large positive, 
while preserving the constant $M=m_1+m_2$.  Moreover, at the collision 
$x_1 \rightarrow x_2$, $\beta_- - \beta_+ e^{2 \nu (x_2 - x_1)} \rightarrow 1$ (as $\beta_- - \beta_+ = 1$), and thus we have that the $z$ dependent off-diagonal terms involve:
$$ m_1 m_2 (1 - e^{2 \nu (x_1 - x_2)}) \rightarrow C_2. $$
Therefore at the collision, we have:

$$A(z) \rightarrow \begin{bmatrix} \beta_- (z^2 + zM + C_2) & \beta_+ (z M e^{-2 \nu x_1} + C_2 e^{-2 \nu x_1})  \\ \beta_-(-zMe^{2 \nu x_1} - C_2e^{2\nu x_1}) & \beta_+ (z^2 - zM - C_2) \end{bmatrix}$$
$$ = \begin{bmatrix} \beta_- (z^2 + zM + C_2) & \beta_+ M e^{-2 \nu x_1}(z + \frac {C_2} M)  \\ -\beta_- Me^{2 \nu x_1}(z + \frac {C_2} M) & \beta_+ (z^2 - zM - C_2) \end{bmatrix}.$$
We see that the off diagonal terms have the same zeros, so we will conjugate the matrix by $D = \begin{bmatrix} \frac 1 {z + \frac {C_2}{ M}} & 0 \\ 0 & z + \frac {C_2}{ M}\end{bmatrix}$.  This singular automorphism of $E$ has an $\mathcal O(-1)$ section and a $\mathcal O(1)$ section, and thus has transformed the collision Higgs field $A$ into one for the bundle $E=\mathcal O(1)\oplus\mathcal O(-1)$.  We are now precisely at the (twisted) theta divisor in the fibre, as claimed earlier.  The new form of the Higgs field is
$$A' = D^{-1} A D = \begin{bmatrix} \beta_- (z^2 + zM + C_2) & \beta_+ M e^{-2 \nu x_1}(z + \frac {C_2} M)^3  \\ -\frac {\beta_- Me^{2 \nu x_1}}{z + \frac {C_2} M} & \beta_+ (z^2 - zM - {C_2}) \end{bmatrix}.$$
Let $P = \begin{bmatrix} a & b \\ 0 & d\end{bmatrix}$ be a gauge transformation of $\mathcal O (1) \bigoplus \mathcal O (-1)$, meaning $a,d$ are numbers and $b$ is a polynomial with $\deg b = 2$. Then conjugating $A'$ by $P$ gives:
$${A}'' = P^{-1} A' P = \begin{bmatrix} a^{-1} & -\frac b {ad} \\ 0 & d^{-1} \end{bmatrix}  \begin{bmatrix} \beta_- (z^2 + zM + C_2) & \beta_+ M e^{-2 \nu x_1}(z + \frac {C_2} M)^3  \\ -\frac {\beta_-Me^{2 \nu x_1}}{z + \frac {C_2} M} & \beta_+ (z^2 - zM - C_2) \end{bmatrix} \begin{bmatrix} a & b \\ 0 & d \end{bmatrix} $$
$$ = {\tiny \begin{bmatrix} \beta_- (z^2 + zM + C_2) + \frac {\beta_- Me^{2 \nu x_1}}{z + \frac {C_2} M} \frac b d & \frac b a (\beta_- (z^2 + zM + C_2) + \frac {\beta_-Me^{2 \nu x_1}}{z + \frac {C_2} M} \frac b d - \beta_+ (z^2 - zM - C_2)) + \frac d a \beta_+ M e^{-2 \nu x_1}(z + \frac {C_2} M)^3 \\  -\frac {\beta_-Me^{2 \nu x_1}}{z + \frac {C_2} M} \frac a d &  -\frac {\beta_- Me^{2 \nu x_1}}{z + \frac {C_2} M} \frac b d +  \beta_+ (z^2 - zM - C_2)\end{bmatrix}}.$$
Choosing $\frac a d = - \frac 1 {\beta_- M e^{2\nu x_1}}$, and $b = \beta_- M a(z+\frac {C_2} M)^2$ we have:
$$ {A}'' = \begin{bmatrix} \beta_- z^2 & \beta_- M (z + \frac {C_2} M)^2(-z^2 + \beta_+(zM + C_2)) - \beta_+ \beta_-M^2(z + \frac {C_2} M)^3\\  \frac 1 {z + \frac {C_2} M} &  \beta_+ z^2 + zM + C_2\end{bmatrix}.$$
With further simplification we have:
$$ {A}'' = \begin{bmatrix} \beta_- z^2 & -\beta_- M (z + \frac {C_2} M)^2 z^2 \\  \frac 1 {z + \frac {C_2} M} & \beta_+ z^2 + zM + {C_2}\end{bmatrix}.$$

We see that ${A}''$ is determined by $\beta_+, \beta_-, M, C_2$, which are all determined by $\Tr A(z)$. Therefore any Higgs fields with the same trace can be conjugated to the same form ${A}''$ at the collision $x_1 = x_2$, and so we have the uniqueness of the collision point, up to gauge.

We also note, somewhat in contrast to \cite{SR:13}, that the Higgs field $A''$ has a pole of order $1$ at $z=-C_2/M$ on $\mathbb P^1$, and so is not purely holomorphic on this chart.  This originates in the fact that the off-diagonal terms of the original solution $A$ for $E=\mathcal O\oplus\mathcal O$ were of degree strictly less than $2$.  Hence, the collision necessitates not only a change of bundle type but also a change of Higgs field type, that is, to a parabolic Higgs field with and order $1$ pole.

\begin{remark} The case of $d=2$ was also investigated by N. Hitchin recently in \cite{NJH:17} in the context of Nahm's equations, where the singular Hitchin fibres are studied and non-classical conserved quantities are shown to exist.\end{remark}

\section{Inverse Problem: Recovering $m_j$, $x_j$}

Even though we already know by \autoref{thm:linCF} that the dynamics 
linearize on the Jacobian $J^{\tilde g-1}$ of the Riemann surface $Y$, 
the actual task is to solve the peakon ODE system \eqref{eq:peakonODEs}.    
We will proceed in the following way.  First, 
 we will study the eigenvector mapping 
\begin{equation}\label{eigvecmap}
\phi_t: Y\rightarrow \mathbb P(E) 
\end{equation}
with $E=\mathcal O\oplus\mathcal O$, given by 
the eigenvalue problem: 
\begin{equation}\label{eigenvaluep}
A(z) v=w v. 
\end{equation}
For reasons of symmetry we will shift $w \rightarrow w+\tfrac12 \Tr (A(z))$.  
This is equivalent to making $A$ traceless, which corresponds to passing to the $\mbox{SL}(2,\mathbb C)$ Hitchin moduli space, which has fibres of the same dimension $d-1$ over a smaller base.  Now, we will construct a meromorphic function $W$ on $Y$ by 
taking the ratio of two holomorphic sections of the line bundle $\phi_t^*\mathcal{O}_{\mathbb P(E)}(1)\in J^{0}(Y)$, where $\mathcal{O}_{\mathbb P(E)}(1)$ is the hyperplane bundle.  Recall that the off-diagonal polynomials of $A$ are generically of 
degree $d-1$ (see \autoref{thm:Astructure}).  
Indeed, by expanding \eqref{eigenvaluep} with shifted $w$, we get  
\begin{align}
A_{11}(z) \zeta _1 +A_{12}(z)\zeta _2=&(w+\tfrac12 \Tr A(z))\zeta _1 \label {firsteig}\\
A_{21}(z) \zeta _1 +A_{22}(z)\zeta _2 =&(w+\tfrac12 \Tr A(z)) \zeta _2  \label{secondeig},
\end{align}
 which results in three useful formulas for the 
\textit{ generalized Weyl function}  $W=\frac{\zeta _2}{\zeta_1}$: 
\begin{align}
W&=\frac{w -\tfrac12 \Tr (A(z)\sigma_3)}{A_{12}(z)}\label{Weylanalog1}, \\
W&=\frac{A_{21}(z)}{w +\tfrac12 \Tr (A(z) \sigma_3)}\label{Weylanalog2}, \\
W&=\frac{A_{21}(z) +
A_{22}(z) W }{A_{11}(z) + A_{12}(z) 
W}, \label{Weylanalog3}
\end{align}
where $\sigma_3 =\begin{bmatrix}  1&0\\0& -1 \end{bmatrix}$.  
Thus $W$ can be viewed as a ratio of two holomorphic functions on an affine chart on the sphere, 
one with zeros at the zeros of $A_{12}$ , the other with zeros
at the zeros of $A_{21}$, properly lifted to $Y$ (i.e. 
if $A_{12}(z_0)=0$ then the unique lift is  $(z_0, w_0=-\tfrac12 \Tr ( A(z_0)\sigma_3)$). The numbers of zeros and poles are equal, and so we view this as a meromorphic section of a holomorphic line bundle $U$ of degree $0$ on $Y$.

The curious reader may wonder how this squares up with the degree of the line bundle in the preceding section.  For $d=2$, the line bundle $S$ on $Y$ had degree $2$.  This is again precisely the difference between working in $J^0$ and $J^2$, the translation of which was achieved by twisting by the pullback of $\mathcal O(1)$, the generator of the Picard group of the projective line.  In general, the actual line bundle $S$ that pushes forward to reconstruct the Higgs bundle on $\mathbb P^1$ is, by Grothendieck-Riemann-Roch, of degree $d$ and so the two natural Jacobians are $J^{d}$ and $J^{0}$, the former being the home of the actual spectral line bundle $S$ and the latter being the home of the line bundle $U$ whose section is the above ratio.  If $d$ is odd, then the passage from $J^0$ to $J^d$ involves an additional line bundle $R$ of degree $1$ on $Y$ that completes the equivalence $S=U\otimes\pi^*\mathcal O(1)^{(d-1)/2}\otimes R$.  In the context of the preceding section, collisions of particles occur when zeros of the numerator and denominator of the meromorphic section align.  Once all of the zeros upstairs line up with all of those downstairs, we are now at a line bundle with a constant holomorphic section --- in other words, the trivial line bundle $\mathcal O_Y$.  This point will correspond (via twisting to degree $\tilde g-1=d-2$) to a line bundle in the theta divisor.  In general, the emerging picture is that $\Theta$ is stratified by different pairings of collisions, which in turn correspond to pairings of zeros and poles in $W$.

The three formulas for $W$ above are also quite useful insofar as they reveal different 
aspects of going to the limit $\beta_+=0$.  We recall that 
this limit corresponds to the pure Camassa-Holm (CH) peakons ($\beta_-=1$) 
for which the inverse formulas yielding masses and positions 
exist \cite{beals-sattinger-szmigielski:stieltjes}.  To make this connection we observe that, if $\beta_+\rightarrow 0$, then by \eqref{Weylanalog3} one automatically obtains 
\begin{equation*}
W=\frac{A_{21}(z)}{A_{11}(z)}\stackrel{\autoref{cor:trueLax}}{=}\frac{T_{21}(\lambda)}{T_{11}(\lambda)}, \qquad \lambda =\frac1z, 
\end{equation*}
which is the desired result.  If, on the other hand, one uses 
\eqref{Weylanalog1} then one realizes that  in the limit the spectral curve  becomes $w=\tfrac12 A_{11}(z)$ and the genus drops to $0$.   Either way one 
obtains the original Weyl function of the CH peakon problem and $A_{11}(z)$ is the desired spectral invariant.   

After shifting $w$ the algebraic curve $C(w,z)$ reads
\begin{equation}\label{surfaceeq}
w^2=P(z) , \quad P=\tfrac 14 (\Tr A(z))^2 -\det \beta z^{2d}.  
\end{equation}
It is elementary to check, in view of $\beta_--\beta_+=1$,  that 
$P(z)=z^{2d}+O(z^{2d-1}), \, z\rightarrow \infty$.  Moreover, in the limit $\beta_+\rightarrow 0$, $P(z)$ becomes a perfect square
as we indicated earlier.   

  Since we will need a bit of information 
about this surface let us review some elementary facts about this surface.  
Assuming for simplicity that all roots of $P(z)$ are simple we can write 
\begin{equation}\label{canonicalcurve}
w^2=\prod _{j=1}^{2d}(z-z_j), \quad 
z_i\ne z_j, \quad i\ne j.  
\end{equation}
The curve $Y$ resulting from the compactification of $C(w,z)=0$ has a branch points at $z_j$.  We will define the upper 
sheet as the one on which $w=+\sqrt{P}$ and the value of 
the right hand side being defined as positive for large positive values of 
of $z$, with $w=-\sqrt{P}$ on the lower sheet.  
Let us denote by $C_d$ the coefficient of $z^0$ in $\Tr A(z)$.  Observe 
that if $C_d\neq 0$ then $z=0$ is not a branch point and thus we 
have two lifts $\pi^{-1} (0)$ on $Y$ to be denoted 
$0 ^+$ on the upper sheet, $0^-$ on the lower sheet respectively.

The following theorem is crucial for finding a solution to the inverse problem:

\begin{theorem} \label{thm:PadeApprox}
Suppose $C_d =\Tr A(0)>0$, then on 
the upper sheet of $Y$, using $z$  as a local parameter around  $0^+$, we have 
\begin{equation}\label {zWeylassympt}
W-\frac{A_{21}(z)}{A_{11}(z)}=\textrm{O}(z^{2d}), \quad z \rightarrow 
0.  
\end{equation}
The same result holds on the lower sheet if $C_d<0$.  
\end{theorem}

\begin{proof} 
Assume $C_d >0$.  Using \eqref{Weylanalog2} we see 
\begin{equation*} 
W-\frac{A_{21}(z)}{A_{11}(z)}=
\frac{A_{21}(z)}{A_{11}(z)}\frac{(-w+\tfrac12 \Tr A(z))}{(w+\tfrac12 
\Tr (A(z) \sigma3))}.
\end{equation*}
Multiplying top and bottom of the right hand side by 
$(w+\tfrac12 \Tr A(z))$ and using the equation determining 
the surface $Y$ we get
\begin{equation*} 
W-\frac{A_{21}(z)}{A_{11}(z)}=
\frac{A_{21}(z)}{A_{11}(z)}
\frac{\det \beta z^{2d}}{\big((w+\tfrac12 A_{11}(z))^2-\tfrac14 A^2_{22}(z)\big)}.
\end{equation*}
We now observe that by 
\autoref{thm:Astructure} the first term $\frac{A_{21}(z)}{A_{11}(z)}$ has 
a limit as $z\rightarrow 0$.  We claim that, under the condition that 
$C_d>0$, the denominator $ (w+\tfrac12 A_{11}(z))^2-\tfrac14 A^2_{22}(z)\big)$ has a nonzero limit equal $A_{11}(0)C_d$.  That $A_{11}(0)\neq 0$ follows from  $C_d=f_{[d]}m_{[d]}(\beta_--\beta_+e^{2\nu(x_d-x_1)})=A_{11}(0)(1-\frac{\beta_+}{\beta_-} e^{2\nu(x_d-x_1)})$.  To prove the actual 
claim we observe that $0^+=(\tfrac12 (A_{11}(0)+A_{22}(0)), 0)$ hence  
on the upper sheet  
\begin{equation*} 
(w+\tfrac12 A_{11}(z))^2-\tfrac14 A^2_{22}(z)\rightarrow (A_{11}(0)+\tfrac12 A_{22}(0))^2-\tfrac14 A^2_{22}(0)=A_{11}(0)(A_{11}(0)+A_{22}(0)), 
\end{equation*} 
as claimed.  
 In the case of $C_d<0$, we go to the lower sheet  where this time $0^-=(\tfrac12 (A_{11}(0)+A_{22}(0)), 0)$ and the rest of the computation goes through identically.  
\end{proof} 
For the sake of comparison with \cite{beals-sattinger-szmigielski:stieltjes}
we will also state the previous theorem using a different local 
parametrization $\lambda=\frac{1}{z}$ and write the result 
in terms of $T_{11}(\lambda)$ and $T_{21}(\lambda)$.  

\begin{corollary} 
Suppose $C_d =\Tr A(0)>0$, then on 
the upper sheet of $Y$, using $\frac 1\lambda$  as a local parameter around  $0^+$, we have 
\begin{equation}\label {eq:lambdaWeylassympt}
W-\frac{T_{21}(\lambda)}{T_{11}(\lambda)}=\textrm{O}(\frac{1}{\lambda^{2d}}), \quad \lambda \rightarrow 
\infty
\end{equation}
The same result holds on the lower sheet if $C_d<0$.  
\end{corollary} 
The next theorem gives a complete solution of the inverse 
problem of determining $m_j, x_j$ from the knowledge of $W$:  
\begin{theorem} 
Suppose $Y$ is given by \autoref{surfaceeq} and $C_d>0$.  Let $W$ be the generalized Weyl function.   Then the first
$2d$ terms of the Taylor expansion of $zW$ around $z=0$ 
on the upper sheet determine uniquely 
the peakon measure $m$, given by \eqref{eq:mpeakon}, and thus the field $A(z)$.   
 
Likewise, if $C_d<0$, then we have the analogous result provided the expansion is computed on the lower sheet of 
$Y$.  
\end{theorem} 

\begin{proof} (Sketch)
It is clear that knowing the peakon measure $m$ we 
can define the field $A(z)$  as in Theorem \ref{thm:Astructure}.  
By a minor modification of the arguments in \cite{beals-sattinger-szmigielski:stieltjes, beals-sattinger-szmigielski:moment} the 
operator $L(\lambda)=D^2-\nu ^2 -2\nu \lambda m$ (see \autoref{sec:CF}) 
is unitarily equivalent to $\tilde L(\lambda)=D^2 -2\nu \lambda \tilde m$ 
defined on the finite interval $[-\frac {1}{2\nu}, \frac {1}{2\nu}]$, where the measure $\tilde m$ is related to $m$ in a simple way.  
Thus without any loss of generality we can 
assume that $L(\lambda) $ in our initial problem is stated for $ D^2 -
2\nu\lambda m$.  This problem has a very natural interpretation, namely 
it describes a classical inhomogeneous string of length $\frac1\nu$ with 
discrete mass density $m$ (see \cite{gantmacher-krein,dym-mckean:gaussian}), and, consequently, the inverse problem can 
be solved using the formulas 
originally obtained for Stieltjes' continued fractions \cite{stieltjes}, as explained in \cite{beals-sattinger-szmigielski:stieltjes}, except that the 
Weyl function in our case is not a rational function on 
$\mathbf{P} ^1$ but a meromorphic function (or section) on $Y$.

Recall the main premise of Stieltjes' work.  Given an asymptotic 
expansion of a function 
\begin{equation}\label{eq:fas}
f(\l)=\sum_{i=0}(-1)\frac{c_i}{\l^{i+1}}, \qquad \lambda \rightarrow \infty
\end{equation}
one associates to it a sequence $\{f_j, j\in \N\}$ of continued fractions, which are Pad\'{e} approximants when written as rational functions: 
\begin{equation}\label{confrac}
f_j=\cfrac{1}{\lambda\,a_{1}+\cfrac{1}{a_{2}+\cfrac{1}{\lambda\,a_{3}+\dots
+\cfrac{1}
{a_{j-1}+\cfrac{1}{\lambda\,a_{j}}}}}}
\end{equation}
These are obtained by requiring that each $f_j$ approximates $f$ with an error term being $O(\frac{1}{\l^{j+1}})$.  
In other words, the $j$th approximant recovers $j$ terms of the 
asymptotic expansion of $f$.  The essential role in determining 
how the coefficients $a_1, a_2, \cdots $ are constructed 
from the asymptotic expansion is played 
by the infinite Hankel matrix $H$ constructed out of the 
coefficients of the expansion \eqref{eq:fas} via 
\begin{equation*} 
H=\begin{bmatrix} c_0&c_1&c_2&c_3&\cdots\\c_1&c_2&c_3&c_4&\cdots\\
c_2&c_3&c_4&c_5&\cdots\\
\vdots&\vdots&\vdots&\vdots&\vdots \end{bmatrix}
\end{equation*} 
and the determinants of the $k\times k$ submatrices of $H$ whose top-left entry is $c_\ell$ in the top row of $H$, where $\ell\geq0,k>0$.  We denote this determinant by $\Delta_k^\ell$ and impose the convention that $\Delta^\ell_0=1$.  Using this notation, the formulas for $a_j$ in \eqref{confrac} 
read 
\begin{equation} \label{eq:Sformulas} 
a_{2k}=\frac{(\Delta^0_k)^2}{\Delta_k^1 \Delta_{k-1}^1}, \qquad \qquad 
a_{2k+1}=\frac{(\Delta^1_k)^2}{\Delta_k^1 \Delta_{k+1}^1}.  
\end{equation}

We note that in this setup the function $f$ vanishes 
at $\lambda=\infty$.  Thus to compare with our case we  
might consider $zW=\frac1\lambda W$.  Then \autoref{eq:lambdaWeylassympt} reads 
\begin{equation*}
\frac{W}{\l} -\frac{T_{21}(\lambda)}{\l T_{11}(\lambda)}=\textrm{O}(\frac{1}{\lambda^{2d+1}}), \quad \lambda \rightarrow 
\infty. 
\end{equation*}
This result shows that the first $2d$ terms in the asymptotic --- in fact analytic --- expansion of 
$\frac{W}{\l}$ and $\frac{T_{21}(\lambda)}{\l T_{11}(\lambda)}$ are 
identical.  However, the full continued fraction expansion of 
$\frac{T_{21}(\lambda)}{\l T_{11}(\lambda)}$ is known (see \cite{beals-sattinger-szmigielski:stieltjes}) to be 
\begin{equation*}
\frac{T_{21}(\lambda)}{\l T_{11}(\lambda)}=\cfrac{1}{\lambda\,l_{d}+\cfrac{1}{m_{d}+\cfrac{1}{\lambda\,l_{d-1}+\dots
+\cfrac{1}
{m_{1}+\cfrac{1}{\lambda\,l_0}}}}} 
\end{equation*}
where $l_j=x_{j+1}-x_{j}, j=0, \cdots, d$ and $x_0= -\frac{1}{2\nu}, x_{d+1}=\frac{1}{2\nu}$.  
In our case, we have $2d$ terms which give us the first $(2d-1)$-th 
approximants by \autoref{eq:Sformulas}, while we formally need to go up to the $(2d+1)$-th approximant to determine $m_1$ and $l_0$.  
However, recall that we know that the total length of the string is 
$ \frac{1}{\nu}$; hence, $l_0+=\frac{1}{\nu}-(l_1\cdots +l_d)$.  Likewise, 
the total mass $M$ is known from the coefficient of $\Tr A(z)$ at $z^{d-1}$; hence, $m_1=M-(m_d+\cdots+m_2)$.  
\end{proof} 

\begin{remark} 
The inverse problem for the CF equation \eqref{eq:peakonODEs} 
with $u$ satisfying \eqref{eq:um} (in the special case $C=0$) 
is given a different treatment in \cite{beals-sattinger-szmigielski:CF}.  
The solution outlined in the present paper is closer in spirit to the original 
solution of the inverse problem for the case $\beta_+=0$ in \cite{beals-sattinger-szmigielski:moment}.  

  The outstanding 
problem for either treatment is to characterize the generalized 
Weyl functions $W$ with the known theta function representation of 
the flow (see \cite[Theorem 6.2]{beals-sattinger-szmigielski:CF}).  For example, a function $f(z)$ is a Weyl function for the CH peakon problem ($\beta_+=0$) 
if and only if $f$ is 
\begin{enumerate} 
\item a rational function vanishing at infinity with simple poles on the real axis; 
\item all its residues are  positive.  

\end{enumerate} 
 It remains an open question what replaces these conditions in the CF case.  

\end{remark}

\bibliographystyle{abbrv}
\bibliography{CF}

\end{document}